\documentclass[conference]{IEEEtran}

\usepackage{amsmath,amssymb,amsthm,amsfonts}
\usepackage{mathrsfs}
\usepackage{graphicx}
\usepackage{subfigure}
\usepackage{color}
\usepackage{bbm}

\newcommand{\ie}{i.e.}
\newcommand{\eg}{e.g.}
\newcommand{\etal}{\emph{et al.}}

\newcommand{\mc}[1]{\mathcal{#1}}

\newcommand{\Bs}[1]{\boldsymbol{#1}}

\newcommand{\rspan}[1]{\left\langle #1 \right\rangle}

\newcommand{\uni}[1]{\mathsf{Uni}\left( #1 \right)}

\newcommand{\ter}{\mathsf{T}}
\newcommand{\alice}{\mathsf{A}}
\newcommand{\bob}{\mathsf{B}}
\newcommand{\calvin}{\mathsf{C}}
\newcommand{\eve}{\mathsf{E}}

\newcommand{\Fbb}{\mathbb{F}}

\newcommand{\Kca}{\mathcal{K}}

\newtheorem{theorem}{Theorem}
\newtheorem{corollary}{Corollary}
\newtheorem{lemma}{Lemma}
\newtheorem{definition}{Definition}
\newtheorem{proposition}{Proposition}
\newtheorem{example}{Example}

\newenvironment{remark}{\noindent\textbf{Remark:}}{}

\title{Multi-terminal Secrecy in a \\
Linear Non-coherent Packetized  Networks}
\author{
\IEEEauthorblockN{Mahdi~Jafari~Siavoshani}
\IEEEauthorblockA{Ecole Polytechnique F\'{e}d\'{e}rale de Lausanne\\
Email: mahdi.jafarisiavoshani@epfl.ch}
\and
\IEEEauthorblockN{Christina~Fragouli}
\IEEEauthorblockA{Ecole Polytechnique F\'{e}d\'{e}rale de Lausanne\\
Email: christina.fragouli@epfl.ch}
}

\begin{document}
\maketitle

\begin{abstract}
We consider a group of $m+1$ trusted nodes that aim
to create a shared secret key $\mc{K}$ over a network in the
presence of a passive eavesdropper, Eve. We assume a linear non-coherent 
network coding broadcast channel (over a finite field $\Fbb_q$) 
from one of the honest nodes (\ie, Alice) to the rest of them including
Eve. All of the trusted nodes can also discuss over a cost-free
public channel which is also overheard by Eve.

For this setup, we propose upper and lower bounds for the secret key 
generation capacity assuming that the field size $q$ is very large. 
For the case of two trusted terminals ($m=1$) our upper and lower
bounds match and we have complete characterization for the secrecy 
capacity in the large field size regime.
\end{abstract}

\section{Introduction}

For communication over a network performing linear network coding, 
Cai and Yeung \cite{CaiYeu-ISIT02-SecNetCode} 
introduced the problem of securing a multicast transmission 
against an eavesdropper. In particular, consider a network implementing 
linear network coding over a finite field $\Fbb_q$. Let us assume that the min-cut
value from the source to each receiver is $c$. From the main theorem of 
network coding \cite{AhlCaiLiYeu-IT00-NetCode,LiYeCa-IT03-NetCodeLin}
we know that a source 
can send information at rate equal to the min-cut $c$ to the destinations, in the 
absence of any malicious eavesdropper. Now, suppose there is a passive eavesdropper, Eve, 
who overhears $\rho$ arbitrary edges in the network. The \emph{secure network coding} 
problem is to design a coding scheme such that Eve does not obtain any 
information about the messages transmitted from the source to destinations.
Cai and Yeung \cite{CaiYeu-ISIT02-SecNetCode} showed that the secrecy capacity for
this problem is $c-\rho$ and can be achieved if the field size $q$ is sufficiently
large. Later this problem formulation has been investigated in many other works.
Feldman \etal~\cite{FeldMalSteSer-ALLERTON04-SecNetCode} showed that by sacrificing a small
amount of rate, one might find a secure scheme that requires much smaller field size.
Rouayheb \etal~\cite{RouSolj-ISIT07-WiretapNetII} observed that this problem 
can be considered as a generalization of the Ozarow-Wyner wiretap channel of type II.
Silva \etal~\cite{SilKsch-IT11-UnivSecNetCode} proposed a universal coding scheme 
that only employs encoding at the source. 

In contrast to the previous work, in this paper we study the problem of
secret key sharing among multiple terminals when nodes can send feedback
over a public channel. We consider a source multicasting information over 
a network at rate equal to the min-cut $c$ to the destinations. We also assume that the 
relay nodes in the network perform linear randomized network coding which is modeled by a non-coherent 
transmission scheme.
Motivated by \cite{JaMoFrDi-IT11,SiKschKoet-IT10-MatrixChnl}, we model a 
non-coherent network coding scenario by a multiplicative matrix channel over 
a finite field $\Fbb_q$ with uniform and i.i.d. distribution over transfer 
matrices in every time-slot.

The problem of key agreement between a 
set of terminals with access to noisy broadcast channel and public discussion 
channel (visible to the eavesdropper) was studied in \cite{CsNa-IT08-ChnlSecrecy}, where some 
achievable secrecy rates were established,
assuming Eve does not have access to the noisy broadcast
transmissions. This was generalized in \cite{GohAna-IT10-P1,GohAna-IT10-P2} by developing
(non-computable) outer bounds for secrecy rates.
However, to the best of our knowledge, ours is the first work to consider multi-terminal
secret key agreement over networks employing randomized network coding, 
when a passive eavesdropper has access to the broadcast transmissions.

Our contributions in this paper are as follows.
For the secret key sharing problem introduced above, we propose an 
asymptotic achievability scheme assuming that the 
field size $q$ is large. This scheme is based on \emph{subspace coding} 
and can be extended for arbitrary number of terminals. Using the result 
of \cite{CsNa-IT08-ChnlSecrecy}, we derive an upper bound for this 
problem. For $m=1$, the proposed lower bound matches the upper bound 
and the \emph{secret key generation capacity} is characterized.
However, for $m\ge 2$, depending on the channel parameters, the upper
and lower bound might match or not.

The paper is organized as follows.
In \S\ref{sec:NonCohSecrecy-NotatioSetup} we introduce our notation and
the problem formulation and present some preliminaries. 
In \S\ref{sec:IndpBrdCstChnl-Secrecy-UpperBound}, 
we state a general 
upper bound for the key generation capacity and evaluate it for the 
non-coherent network coding broadcast channel. The main results of the
paper are presented in \S\ref{sec:NonCohChnl-Secrecy-AchvScheme}.

\section{Notation and Setup}\label{sec:NonCohSecrecy-NotatioSetup}
\subsection{Notation}
We use $\rspan{X}$ to denote the row 
span of a matrix $X$. We use also $[i:j]$ to 
denote $\{i,i+1,\ldots,j\}$ where $i,j\in\mathbb{Z}$.

Let $\Pi$ be an arbitrary vector space of finite dimension defined over a finite
field $\Fbb_q$. Suppose $\Pi_1$ and $\Pi_2$ are two subspaces of $\Pi$, \ie,
$\Pi_1\sqsubseteq \Pi$ and $\Pi_2\sqsubseteq \Pi$. We use $\Pi_1\cap \Pi_2$ to denote the common
subspaces of both $\Pi_1$ and $\Pi_2$ and $\Pi_1+\Pi_2$ as the smallest subspace that contains
both $\Pi_1$ and $\Pi_2$. Two subspaces $\Pi_1$ and $\Pi_2$ are called \emph{orthogonal}
if $\Pi_1\cap \Pi_2 =\{\Bs{0}\}$. Two subspaces $\Pi_1$ and $\Pi_2$ of $\Pi$
are called \emph{complementary} if they are orthogonal and $\Pi_1+\Pi_2=\Pi$.

Now, consider two subspaces $\Pi_1$ and $\Pi_2$. We define the subtraction
of $\Pi_2$ from $\Pi_1$ by $U=\Pi_1\setminus_s \Pi_2$ where $U$ is any subspace
of $\Pi_1$ which is complementary with $\Pi_1\cap \Pi_2$. Note that, given
$\Pi_1$ and $\Pi_2$, $U$ is not uniquely defined.

For notational convenience, when $\mc{J}$ is a set, by $\Pi_\mc{J}$ we mean
$\Pi_\mc{J}\triangleq \cap_{i\in\mc{J}} \Pi_i$. 

\subsection{Preliminaries}\label{sec:Preliminaries}

\begin{definition}
We define $\mc{S}(\ell,k)$ to be the set of all subspaces 
of dimension at most $k$ in the $\ell$-dimensional space $\Fbb_q^\ell$.
\end{definition}

\begin{definition}[{see \cite{JaMoFrDi-IT11}}]\label{eq:Definition-xi}
We denote by $\xi(n,d) $ the number of different $n\times \ell$ matrices with
elements from a finite field $\Fbb_q$, such that their rows span a specific
subspace $\pi_d\sqsubseteq\Fbb_q^\ell$ of dimension $d$ where $0\le d\le \min[n,\ell]$.
By using \cite[Lemma~2]{JaMoFrDi-IT11}, $\xi(n,d)$ does not depend 
on $\ell$ and depends on $\pi_d$ only through its dimension $d$.
\end{definition}

\begin{lemma}\label{lem:Uni-and-Joint-k-RandSubSpace}
Suppose that $k$ subspaces $\Pi_1,\ldots,\Pi_k$, with dimensions 
$d_1,\ldots,d_k$, are chosen uniformly at random from $\Fbb_q^n$.
Then w.h.p. (with high probability)\footnote{During the paper by 
``high probability'' we mean probability of order $1-O(q^{-1})$
unless otherwise stated.}
we have
\begin{align*}
\dim\left(\Pi_1+\cdots+\Pi_k\right) = \min\left[d_1+\cdots+d_k,n \right], \quad\text{and} \\
\dim\left(\Pi_1\cap\cdots\cap\Pi_k\right) = \left[d_1+\cdots+d_k - (k-1)n \right]^+.
\end{align*}
Note that even if one of the subspaces, for example $\Pi_1$, is a fixed 
subspace, then the above results are still valid.
\end{lemma}

\begin{proof}
These results follow from \cite[Corollary~1]{JaFrDi-IT12} by using induction
on the number of subspaces.
\end{proof}

\subsection{Problem Statement}
\label{sec:SecShar-NonCohNetCode-ProbStatement}
We consider a set of $m+1 \ge 2$ honest nodes, $\ter_0, \ldots, \ter_{m},$ ($\ter$ stands for ``terminal'') 
that aim to share a secret key $\Kca$ among themselves while keeping 
it concealed from a passive adversary, Eve.  Eve does 
not perform any transmissions, but is trying to eavesdrop on (overhear) 
the communications between the honest nodes. For convenience, sometimes 
we will refer to node $\ter_0,\ter_1,\ter_2,\ldots,$ as ``Alice,'' ``Bob,'' ``Calvin,''
and so on.

We assume that there exists a non-coherent network coding broadcast channel
(which is going to be defined more precisely in the following) from Alice to the other
terminals (including Eve). Also we assume that the legitimate terminals can publicly
discuss over a noiseless rate unlimited public channel. 


Consider a non-coherent linear network coding communication scenario 
where at every time-slot $t$ Alice (terminal $\ter_0$) injects a set of $n_\alice$ vectors 
(packets) of length $\ell$ (over some finite field $\Fbb_q$) into the network, 
denoted by the row vectors of the matrix $X_\alice[t]\in\Fbb_q^{n_A\times \ell}$.
Each terminal $\ter_i$ receives $n_i$ randomly chosen linear 
combinations of the transmitted vectors, namely for 
$r\in\{1,\ldots,m,\eve\}$, we have\footnote{As subscript, we use $i$
to denote for $\ter_i$ for all $i\in[0:m]$.
At some points, we also use
$X_\alice$, $X_\bob$, $X_\calvin$, etc., to denote for $X_0$, $X_1$, $X_2$, etc.}
\begin{equation}\label{eq:NonCohNetCodSecrecy-MatrixChnl}
X_r[t] = F_r[t] X_\alice[t],
\end{equation}
where $F_r[t]\in\Fbb_q^{n_r\times n_\alice}$ is chosen uniformly at random 
among all possible matrices and independently for each receiver
and every time-slot. So for the channel transition probability we can write
\begin{gather}
P_{X_1\cdots X_m X_\eve | X_\alice} (x_1,\ldots,x_m,x_\eve | x_\alice) =  \nonumber\\
P_{X_\eve|X_\alice}(x_\eve|x_\alice) \prod_{i=1}^m P_{X_i|X_\alice}(x_i|x_\alice), \label{eq:NonCohNetCode-IndpChnlTrnsProb}
\end{gather}
where for each $r\in\{1,\ldots,m,\eve\}$ we have (see \cite[Sec~IV-A]{JaMoFrDi-IT11})
\begin{equation*}
P_{X_r|X_\alice}(x_r|x_\alice) \triangleq
\left\{\begin{array}{ll}
q^{-n_r\dim(x_\alice)} & \text{if } \rspan{x_r}\sqsubseteq \rspan{x_\alice},\\
0 & \text{otherwise}.
\end{array} \right.
\end{equation*}
Note that in this setup we do not assume any CSI\footnote{Channel state 
information.} at the transmitter or receivers.

In order to define the secrecy capacity, we use \cite[Definition~1]{JaFrDiPuAr-Asilomar10}
and \cite[Definition~2]{JaFrDiPuAr-Asilomar10} (see also \cite{Ma-IT93,AhlCs-IT93-P1,CsNa-IT08-ChnlSecrecy,GohAna-IT10-P2}).

\section{Upper Bound}\label{sec:IndpBrdCstChnl-Secrecy-UpperBound}
\subsection{Secrecy Upper Bound for Independent Broadcast Channels}
The secret key generation capacity among multiple terminals (without 
eavesdropper having access to the broadcast channel) is completely 
characterized in \cite{CsNa-IT08-ChnlSecrecy}. 
By using this result, it is possible to state an upper bound  for the 
secrecy capacity of the key generation problem among multiple terminals 
where the eavesdropper has also access to the broadcast channel.
This can be done by adding a dummy terminal to the first problem and
giving all the eavesdropper's information to this dummy node and let
it to participate in the key generation protocol. By doing so, the 
secret key generation rate does not decrease.
Hence by combining \cite[Theorem~4.1]{CsNa-IT08-ChnlSecrecy} and
\cite[Lemma~5.1]{CsNa-IT08-ChnlSecrecy}, the following result can be 
stated.

\begin{theorem} \label{thm:SecrecyUpBound-CsNa08}
The secret key generation capacity is upper bounded as follows
\par\nobreak{\small
\setlength{\abovedisplayskip}{-4pt}
\begin{gather*}
C_s\le \nonumber\\
\max_{P_{X_0}} \min_{\lambda\in\Lambda([0:m])} \left[ H(X_{[0:m]}|X_\eve) -\sum_{B\subsetneq [0:m]} \lambda_B H(X_B|X_{B^c},X_\eve) \right],
\end{gather*}
}%
where $\Lambda([0:m])$ is the set of all collections 
$\lambda=\left\{ \lambda_B : B\subsetneq [0:m], B\neq \emptyset \right\}$
of weights $0\le\lambda_B\le 1$, satisfying
\begin{equation*}
\sum_{B\subsetneq [0:m], i\in B} \lambda_B = 1, \quad\quad \forall i\in [0:m].
\end{equation*}
Note that in the above expression for the upper bound, it is possible to 
change the order of maximization and minimization 
\cite[Theorem~4.1]{CsNa-IT08-ChnlSecrecy}.
\end{theorem}

Now, for our problem where the channel from Alice to the other terminals
are assumed to be independent, we can further simplify the upper bound
given in Theorem~\ref{thm:SecrecyUpBound-CsNa08}, as stated in 
Corollary~\ref{cor:SecrecyUpBound-CsNa08}.

\begin{corollary}\label{cor:SecrecyUpBound-CsNa08}
If the channels from Alice to the other terminals are independent, as
described in \eqref{eq:NonCohNetCode-IndpChnlTrnsProb}, then the upper 
bound stated in Theorem~\ref{thm:SecrecyUpBound-CsNa08}, for the 
secret key generation capacity is simplified to
\begin{align}
C_s &\le \max_{P_{X_0}} \min_{j\in[1:m]} I(X_0;X_j|X_\eve) \label{eq:SecrecyUpperBnd-CsNa08-Final-1} \\
 &\le \min_{j\in[1:m]} \max_{P_{X_0}} I(X_0;X_j|X_\eve). \label{eq:SecrecyUpperBnd-CsNa08-Final-2}
\end{align}
\end{corollary}

\begin{proof}
 For the proof please refer to \cite{JafFra-TechRep2012-NonCohNetCodSecrecy}.
\end{proof}

\begin{remark}
Note that \eqref{eq:SecrecyUpperBnd-CsNa08-Final-1} is the best upper
bound one might hope for an independent broadcast channel using
results of \cite{CsNa-IT08-ChnlSecrecy}.
\end{remark}

\begin{remark}
Using \cite[Theorem~7]{Ma-IT93} or \cite[Theorem~2]{AhlCs-IT93-P1}, we observe 
that the bound given in \eqref{eq:SecrecyUpperBnd-CsNa08-Final-2} is indeed 
tight for the two terminals problem where we 
have the Markov chains $X_\bob\leftrightarrow X_\alice \leftrightarrow X_\eve$
(when the channels are independent)
or $X_\alice \leftrightarrow X_\bob \leftrightarrow X_\eve$
(when the channels are degraded). 
\end{remark}

\subsection{Upper Bound for Non-coherent Channel}\label{sec:NonCohChnl-Secrecy-UpperBound}
In the previous section, we have shown that the secret key generation 
rate for our problem can be upper bounded by \eqref{eq:SecrecyUpperBnd-CsNa08-Final-2}.
Now, we need to evaluate the above upper bound for the non-coherent 
network coding channel defined in \S\ref{sec:SecShar-NonCohNetCode-ProbStatement}.

\begin{lemma}\label{lem:Concavity-I(X;Y|Z)-wrt-PX}
For the joint distribution of the form 
\par\nobreak{\small
\setlength{\abovedisplayskip}{-4pt}
\begin{equation*}
P_{X_\alice X_i X_\eve}(x_\alice,x_i,x_\eve)= P_{X_\alice}(x_\alice) P_{X_i|X_\alice}(x_i|x_\alice) P_{X_\eve|X_\alice}(x_\eve|x_\alice)
\end{equation*}}%
the mutual information $I(X_\alice;X_i|X_\eve)$ is a concave function of 
$P_{X_\alice}(x_\alice)$ for fixed $P_{X_i|X_\alice}(x_i|x_\alice)$ and 
$P_{X_\eve|X_\alice}(x_\eve|x_\alice)$.
\end{lemma}

\begin{proof}
For the proof please refer to \cite{JafFra-TechRep2012-NonCohNetCodSecrecy}.
\end{proof}

Similar to \cite[Definition~5]{JaMoFrDi-IT11}, here we define an equivalent 
subspace broadcast channel from Alice (terminal $\ter_0$) to the rest of terminals as follows. 
We assume that Alice sends
a subspace $\Pi_\alice\in\mc{S}(\ell,n_\alice)$ where $\Pi_\alice=\rspan{X_\alice}$ 
and each of the legitimate 
terminals receives $\Pi_i\in\mc{S}(\ell,n_i)$ and Eve receives $\Pi_\eve\in\mc{S}(\ell,n_\eve)$
where $\Pi_i=\rspan{X_i}$ and $\Pi_\eve=\rspan{X_\eve}$, respectively.
The channel transition probabilities are independent and for each receiver $i$ 
is defined as follows
\begin{equation*}
P_{\Pi_i|\Pi_\alice}(\pi_i|\pi_\alice) \triangleq
\left\{\begin{array}{ll}
\xi \big(n_i,\dim(\pi_i) \big) q^{-n_i\dim(\pi_\alice)} & \text{if } \pi_i\sqsubseteq \pi_\alice,\\
0 & \text{otherwise},
\end{array} \right.
\end{equation*}
where the function $\xi$ is defined in Definition~\ref{eq:Definition-xi}.

\begin{lemma}\label{lem:I(X;Y|Z)=I(Pi_X;Pi_Y|Pi_Z)}
For every input distribution $P_{X_\alice}$ there exists an input distribution
$P_{\Pi_\alice}$ such that $I(X_\alice;X_i|X_\eve)=I(\Pi_\alice;\Pi_i|\Pi_\eve)$
and vice-versa.
\end{lemma}

\begin{proof}
For the proof please refer to \cite{JafFra-TechRep2012-NonCohNetCodSecrecy}.
\end{proof}

So by Lemma~\ref{lem:I(X;Y|Z)=I(Pi_X;Pi_Y|Pi_Z)}, in order to maximize
$I(X_\alice;X_i|X_\eve)$ with respect to $P_{X_\alice}$ it is sufficient
to solve an equivalent problem, \ie, maximize $I(\Pi_\alice;\Pi_i|\Pi_\eve)$
with respect to $P_{\Pi_\alice}$; which is seemingly a simpler optimization
problem.

\begin{lemma}\label{lem:I(X;Y|Z)-UnifDistMaximize}
The input distribution that maximizes $I(\Pi_\alice;\Pi_i|\Pi_\eve)$ 
is the one which is uniform over all subspaces having the same dimension.
\end{lemma}

\begin{proof}
By the concavity of $I(\Pi_\alice;\Pi_i|\Pi_\eve)$ with respect to 
$P_{\Pi_\alice}$, that is stated in Lemma~\ref{lem:Concavity-I(X;Y|Z)-wrt-PX}, 
the proof follows by an argument very similar to \cite[Lemma~8]{JaMoFrDi-IT11}.
\end{proof}

\begin{lemma}\label{lem:MultTerSec-NonCohNetCode-UpperBound}
Asymptotically in the field size, we have
\begin{gather*}
\max_{P_{X_\alice}}I(X_\alice;X_i|X_\eve) = \max_{P_{\Pi_\alice}} I(\Pi_\alice;\Pi_i|\Pi_\eve)= \nonumber\\
 \left( \min[n_\alice,n_i+n_\eve]-n_\eve \right) \left( \ell - \min[n_\alice,n_i+n_\eve] \right) \log{q}.
\end{gather*}
\end{lemma}

\begin{proof}
For the proof refer to \cite{JafFra-TechRep2012-NonCohNetCodSecrecy}.
\end{proof}

Thus, by using the upper bound given in \eqref{eq:SecrecyUpperBnd-CsNa08-Final-2} and 
Lemma~\ref{lem:MultTerSec-NonCohNetCode-UpperBound} we have the following
result for the upper bound on the secret key generation rate, 
as stated in Theorem~\ref{thm:MultTerSec-NonCohNetCode-UpperBound}.
\begin{theorem}\label{thm:MultTerSec-NonCohNetCode-UpperBound}
The secret key generation rate in a non-coherent network coding scenario,
which is defined in \S\ref{sec:SecShar-NonCohNetCode-ProbStatement},
is asymptotically (in the field size) upper bounded by
\par\nobreak{\small
\setlength{\abovedisplayskip}{-4pt}
\begin{gather*}
C_s \le \nonumber\\ 
\min_{i\in[1:m]} \big[ \left( \min[n_\alice,n_i+n_\eve]-n_\eve \right) 
\left( \ell - \min[n_\alice,n_i+n_\eve] \right) \big] \log{q}.
\end{gather*}}%
\end{theorem}

\begin{remark}
Note that if $n_\eve=n_\alice$ then the secret key generation rate is zero
because Eve is so powerful that she overhears all of the transmitted
information.
\end{remark}

\section{Asymptotic Achievability Scheme}\label{sec:NonCohChnl-Secrecy-AchvScheme}
Here in this section, we describe our achievability scheme for the secret 
key sharing problem among multiple terminals in a non-coherent network 
coding setup.

Without loss of generality, let us assume that\footnote{If
$n_\alice \ge \ell$ then Alice can reduce the number of injected 
packets into the network from $n_\alice$ to some smaller number 
$n'_\alice$ where $n'_\alice < \ell$.} $n_\alice < \ell$. 
Moreover, in this work we focus on the asymptotic regime where the 
field size is large. Suppose that Alice
broadcasts a message $X_\alice[t]$ at time-slot $t$ of the following
form
\begin{equation}\label{eq:AchvAliceSpecialFormMsg}
X_\alice[t] = \left[\begin{array}{cc}
I_{n_\alice\times n_\alice} & M[t]
\end{array} \right],
\end{equation}
where $M[t]\in\Fbb_q^{n_\alice\times(\ell-n_\alice)}$ is a uniformly at
random distributed matrix. The rest of legitimate terminals and Eve receive 
a linear transformed version of $X_\alice[t]$
according to the channel introduced in \eqref{eq:NonCohNetCodSecrecy-MatrixChnl}.

For each terminal $r\in\{\alice,1,\ldots,m,\eve\}$, we define
the subspace $\Pi_r \triangleq \rspan{X_r}$. Then, for every $r\neq \alice$
we have $\Pi_r\sqsubseteq \Pi_\alice$.
Because of \eqref{eq:AchvAliceSpecialFormMsg}, after broadcasting $X_\alice[t]$, 
the legitimate terminals learn the channel state and reveal the 
channel transfer matrices $F_r[t]$, $r\in[1:m]$, publicly over the 
public channel. 
Thus Alice can also recover the subspaces $\Pi_r$ for all of 
the legitimate terminals.

Now, for each non-empty subset $\mc{J}\subseteq[1:m]$ of legitimate receivers,
let us define the subspace $U_{\mc{J}}$ as follows
\begin{align}\label{eq:Def_U_J}
U_{\mc{J}} \triangleq  
\Pi_{\mc{J}}\setminus_s \left( \sum_{i\in\mc{J}^c} \Pi_{i\mc{J}} + \Pi_{\eve\mc{J}} \right),
\end{align}
where $\Pi_{\mc{J}}=\cap_{i\in\mc{J}} \Pi_i$, $\Pi_{i\mc{J}}=\Pi_i\cap\Pi_{\mc{J}}$, 
and $\Pi_{\eve\mc{J}}=\Pi_\eve\cap\Pi_{\mc{J}}$. By definition, $U_{\mc{J}}$
is the common subspace among the receivers in $\mc{J}$ which is orthogonal
to all of the subspaces of other terminals, \ie, it is orthogonal to
$\Pi_i$, $i\in\mc{J}^c$, and $\Pi_\eve$ (see also Fig.~\ref{fig:NonCohSecrecy-AchvSchmSubspaces}). Note that the subspaces 
$U_\mc{J}$'s are not uniquely defined. However, from the definition
of the operator ``$\setminus_s$'', it can be easily shown that the dimension 
of each $U_\mc{J}$ is uniquely determined and equal to
\begin{equation}\label{eq:U_J-Dimension}
\dim(U_{\mc{J}}) = \dim(\Pi_{\mc{J}}) - \dim\left( \sum_{i\in\mc{J}^c} \Pi_{i\mc{J}} + \Pi_{\eve\mc{J}} \right).
\end{equation}

If Alice had the subspace $\Pi_\eve$ observed by Eve, she would be able 
to construct subspaces $U_\mc{J}$'s; but she does not have $\Pi_\eve$.
However, because the subspaces $\Pi_i$'s and $\Pi_\eve$ are chosen
independently and uniformly at random from $\Pi_\alice$, and because
the field size $q$ is large, Alice, by applying Lemma~\ref{lem:Uni-and-Joint-k-RandSubSpace}, 
can find the dimension of each $U_\mc{J}$ w.h.p. 
Then it can be easily observed that 
(\eg, see \cite[Lemma~3]{JaFrDi-IT12}) if Alice chooses a uniformly 
at random subspace of $\Pi_\mc{J}$ with dimension $\dim(U_\mc{J})$
then it satisfies \eqref{eq:Def_U_J} w.h.p., so it can be a possible candidate for
$U_\mc{J}$.

Now, consider $2^m-1$ different non-empty subsets of $[1:m]$. To each subset 
$\emptyset\neq\mc{J}\subseteq [1:m]$, we assign a parameter $\theta_{\mc{J}}\ge 0$ 
such that the following set of inequalities hold,
\par\nobreak{\small
\setlength{\abovedisplayskip}{-2pt}
\begin{equation}\label{eq:IneqCondTheta}
\theta_{\mc{J}_1}+\cdots+\theta_{\mc{J}_k} \le \dim\left( U_{\mc{J}_1}+\cdots+U_{\mc{J}_k}+\Pi_\eve \right) - \dim(\Pi_\eve),
\end{equation}}%
for any $k\in[1:2^{(2^m-1)}-1]$ and any different selection of subsets 
$\mc{J}_1,\ldots,\mc{J}_k$. 
Note that the right hand side of the inequalities defined in \eqref{eq:IneqCondTheta}
depend on the actual choice of subspaces $U_{\mc{J}}$'s. But, as described above, 
in the following we assume that $U_\mc{J}$'s are chosen uniformly at random from
$\Pi_\mc{J}$.

If Alice knows the subspace $\Pi_\eve$, then we can state the 
following result.

\begin{lemma}\label{lem:ExtractComplementarySubspace-U}
There exists subspaces $U'_{\mc{J}}\sqsubseteq U_{\mc{J}}$ such that
$\dim(U'_{\mc{J}})=\theta_{\mc{J}}$ for all $\emptyset\neq\mc{J}\subseteq [1:m]$, and
$U'_{\mc{J}}$'s and $\Pi_\eve$ are orthogonal subspaces (\ie,
$\dim(\Pi_\eve + \sum_i U'_{\mc{J}_i})= \dim(\Pi_\eve) + \sum_i \theta_{\mc{J}_i}$) if and only if
$\theta_{\mc{J}}$'s are non-negative integers and satisfy 
\eqref{eq:IneqCondTheta}.
\end{lemma}
\begin{proof}
The proof of this lemma is based on \cite[Lemma~4]{KhojKesh-ICC11} and
can be found in \cite{JafFra-TechRep2012-NonCohNetCodSecrecy}.
\end{proof}

Fig.\ref{fig:NonCohSecrecy-AchvSchmSubspaces} depicts pictorially
the relation between subspaces introduced in the above discussions.

\begin{figure}[htb]
\centering
\includegraphics[scale=0.85]{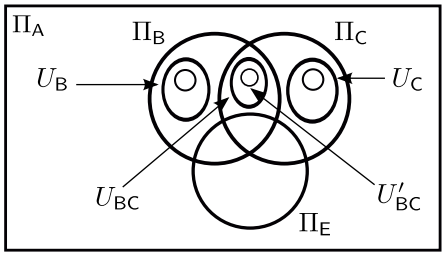}
\caption{The relations between subspaces $\Pi$'s, $U$'s, 
and $U'$'s for the case of $m=2$.}
\label{fig:NonCohSecrecy-AchvSchmSubspaces}
\end{figure}

Although in practice Alice only knows the dimension of $\Pi_\eve$ (w.h.p.),
but still she can find subspaces $U'_{\mc{J}}\sqsubseteq U_{\mc{J}}$
such that the result of Lemma~\ref{lem:ExtractComplementarySubspace-U}
holds w.h.p., as stated in Lemma~\ref{lem:RandomExtractComplementarySubspace-U}.

\begin{lemma}\label{lem:RandomExtractComplementarySubspace-U}
Alice can find subspaces $U'_{\mc{J}}\sqsubseteq U_{\mc{J}}$ such that
$\dim(U'_{\mc{J}})=\theta_{\mc{J}}$ for all $\emptyset\neq\mc{J}\subseteq [1:m]$, 
and $U'_{\mc{J}}$'s are orthogonal subspaces and 
$U'_{\mc{J}}$'s and $\Pi_\eve$ are orthogonal subspaces w.h.p., if and only if
$\theta_{\mc{J}}$'s are non-negative integers and satisfy 
\eqref{eq:IneqCondTheta}.
\end{lemma}
\begin{proof}
For the proof refer to \cite{JafFra-TechRep2012-NonCohNetCodSecrecy}.
\end{proof}

Then, we have the following result.

\begin{theorem}\label{thm:NonCoh-Secrecy-AchvThm}
The secret key sharing rate given by the solution of the following 
convex optimization problem can be asymptotically (in the field size) achieved
\begin{equation*}
\begin{array}{ll}
 \mathrm{maximize} & \left[ \min_{r\in[1:m]} \sum_{\mc{J}\ni r} \theta_{\mc{J}} \right] (\ell - n_\alice)\log{q} \\
\mathrm{subject\ to} & \theta_{\mc{J}}\ge 0, \quad \forall \mc{J}\subseteq [1:m],\ \mc{J}\neq\emptyset, \quad \mathrm{\ and\ }\\
&\theta_{\mc{J}_1}+\cdots+\theta_{\mc{J}_k} \le \\
& \quad\quad \dim\left( U_{\mc{J}_1}+\cdots+U_{\mc{J}_k}+\Pi_\eve \right) - \dim(\Pi_\eve)\\
& \forall k,\ \forall \mc{J}_1,\ldots,\mc{J}_k:\ \emptyset\neq\mc{J}_i\subseteq[1:m],  \\
& \mc{J}_i\neq\mc{J}_j \ \mathrm{if} \ i\neq j, 
\end{array}
\end{equation*}
where for every $\mc{J}$, $U_\mc{J}$ is chosen uniformly at random from $\Pi_\mc{J}$
with the dimension calculated by \eqref{eq:U_J-Dimension} under the assumption
that $\Pi_1,\ldots,\Pi_m$, and $\Pi_\eve$ are selected independently and uniformly at random
from $\Pi_\alice$ with dimensions $n_1,\ldots,n_m, n_\eve$.
\end{theorem}

\begin{proof}[Proof of Theorem~\ref{thm:NonCoh-Secrecy-AchvThm}]
Let Alice use the broadcast channel $N$ times by sending matrices 
$X_\alice[1],\ldots,X_\alice[N]$ of the form \eqref{eq:AchvAliceSpecialFormMsg}. 
As mentioned before, in every time-slot
$t$, each of the legitimate terminals sends publicly the channel transfer 
matrix it has received.

Then, let us define $\hat{\theta}_{\mc{J}} \triangleq \lfloor N\theta_{\mc{J}} \rfloor$
for all $\mc{J}$ and consider the following set of inequalities
\begin{gather}
\hat{\theta}_{\mc{J}_1}+\cdots+\hat{\theta}_{\mc{J}_k} + N\dim(\Pi_\eve) \le \nonumber\\
\dim\left( \bigoplus_{t=1}^N U_{\mc{J}_1}[t] + \cdots + \bigoplus_{t=1}^N U_{\mc{J}_k}[t] 
+ \bigoplus_{t=1}^N \Pi_\eve[t] \right), \label{eq:IneqCondTheta-TimeExt}
\end{gather}
where ``$\oplus$'' is the direct sum operator. Each of 
$\hat{U}_{\mc{J}_i}\triangleq\bigoplus_{t=1}^N U_{\mc{J}_i}[t]$ is a subspace of 
an $N\times n_\alice$ dimensional space $\bigoplus_{t=1}^N \Pi_\alice[t]$.
Similarly, we have $\hat{\Pi}_\eve \sqsubseteq \bigoplus_{t=1}^N \Pi_\alice[t]$
where $\hat{\Pi}_\eve\triangleq\bigoplus_{t=1}^N \Pi_\eve[t]$.
It can be easily seen that if the set of inequalities \eqref{eq:IneqCondTheta} are
satisfied then the set of inequalities \eqref{eq:IneqCondTheta-TimeExt} are 
also satisfied.

Now, by using Lemma~\ref{lem:RandomExtractComplementarySubspace-U}, 
Alice can find a set of orthogonal subspaces $\hat{U}'_{\mc{J}}$ 
with dimension $\hat{\theta}_{\mc{J}}$ (that are also orthogonal 
to $\hat{\Pi}_\eve$ w.h.p.).
By applying Lemma~\ref{lem:DistilSecurePackets-NonCohSetup} (appeared after this theorem),
one would observe that if Alice uses a basis of $\hat{U}'_{\mc{J}}$ 
($\hat{\theta}_\mc{J}$ linear independent vectors from $\hat{U}'_{\mc{J}}$) 
to share a secret key $\mc{K}_\mc{J}$ with all terminals in $\mc{J}$, then
this key is secure from Eve and all other legitimate terminals in 
$\mc{J}^c$ w.h.p.
Using each key $\mc{K}_\mc{J}$, Alice can send a message of size 
$\hat{\theta}_\mc{J} (\ell-n_\alice)\log{q}$ secretly to the 
terminals in $\mc{J}$. In order to share the key $\Kca_\mc{J}$,
Alice sends publicly a set of coefficients for each terminal in $\mc{J}$
so that each of them can construct the subspace $\hat{U}_\mc{J}$ from
their own received subspace. Note that even having these coefficients,
Eve cannot recover any information regarding $\Kca_\mc{J}$ (for more discussion
see \cite{JaFrDiPuAr-Asilomar10}).


Up until now, the problem of sharing a key $\mc{K}$ among legitimate 
terminals have been reduced to a multicast problem where
Alice would like to transmit a message (\ie, the shared key $\mc{K}$) to a set of terminal where
the $r$th one has a min-cut $\sum_{\mc{J}\ni r} \hat{\theta}_{\mc{J}}$.
From the main theorem of network coding (\eg, see 
\cite{AhlCaiLiYeu-IT00-NetCode,LiYeCa-IT03-NetCodeLin,KoMe-TransNet03-AlgApr,FrgSolj-Book07-NetCodeFund}), we know that this
problem can be solved by performing linear network coding where
the achievable rate is as follows
\begin{equation*}
R_s\le \left[ \frac{1}{N} \min_{r\in[1:m]} \sum_{\mc{J}\ni r} \hat{\theta}_{\mc{J}} \right] (\ell - n_\alice)\log{q}.
\end{equation*}
By increasing $N$, the achievable secrecy rate will be arbitrarily
close to
$R_s\le \left[ \min_{r\in[1:m]} \sum_{\mc{J}\ni r} \theta_{\mc{J}} \right] (\ell - n_\alice)\log{q},$
and we are done.
\end{proof}

\begin{lemma}\label{lem:DistilSecurePackets-NonCohSetup}
 Consider a set of $n_\alice$ packets denoted by the rows of a matrix
$X_\alice\in\Fbb_q^{n_\alice\times\ell}$ of the form $X_\alice=[I\quad M]$, 
where $M\sim\uni{\Fbb_q^{n_\alice\times(\ell-n_\alice)}}$.
Assume that Eve has overheard $n_\eve$ independent linear combinations of these packets,
represented by the rows of a matrix $X_\eve\in\Fbb_q^{n_\eve\times\ell}$.
Then for every $k$ packets $y_1,\ldots,y_k$ that are linear combinations
of the rows of $X_\alice$, if the subspace
$\Pi_Y=\rspan{y_1,\ldots,y_k}$ is orthogonal to $\rspan{X_\eve}$
we have
$I(y_1,\dots,y_k;X_\eve)=0$.
\end{lemma}
\begin{proof}
The proof is stated in \cite[Appendix~B]{JafFra-TechRep2012-NonCohNetCodSecrecy}.
\end{proof}


\subsection{Special Case: Achievability Scheme for Two Terminals}
For simplicity and without loss of
generality we assume that $n_\bob\le n_\alice$ and $n_\eve\le n_\alice$.
The key generation scheme starts by Alice broadcasting a message 
$X_\alice[t]$ at time $t$ of the form of \eqref{eq:AchvAliceSpecialFormMsg}. 
Then, Theorem~\ref{thm:NonCoh-Secrecy-AchvThm} states that the secrecy 
rate $R_s$ is achievable if
\begin{equation*}
 R_s \le \left[ \dim(U_\bob+\Pi_\eve)-\dim(\Pi_\eve) \right] (\ell - n_\alice)\log{q},
\end{equation*}
where $U_\bob = \Pi_\bob \setminus_s \Pi_\eve$ (for convenience we have 
replaced $U_{\{\bob\}}$ with $U_\bob$). Because $U_\bob\cap \Pi_\eve=\{\Bs{0}\}$,
we have
\begin{align*}
 R_s &\le \left[ \dim(U_\bob) \right] (\ell - n_\alice)\log{q} \nonumber\\
 &= \left[ \dim(\Pi_\bob)-\dim(\Pi_\bob \cap \Pi_\eve) \right] (\ell - n_\alice)\log{q} \nonumber\\
 &= \left[ n_\bob - (n_\bob+n_\eve-n_\alice)^+ \right] (\ell - n_\alice)\log{q} \nonumber\\
 &= \left[ \min[n_\alice, n_\bob+n_\eve] - n_\eve \right] (\ell - n_\alice)\log{q},
\end{align*}
where this is the same as the upper bound given in 
Theorem~\ref{thm:MultTerSec-NonCohNetCode-UpperBound}. This is obvious
when $n_\alice \le n_\bob+n_\eve$. On the other hand, if $n_\alice>n_\bob+n_\eve$,
then Alice can reduce the number of injected packets in every time-slot from 
$n_\alice$ to $n_\bob+n_\eve$ (there is no need to use more
than $n_\bob+n_\eve$ degrees of freedom).

\begin{remark}
Note that in the above scheme, as long as $n_\eve < n_\alice$, the secrecy 
rate is non-zero.

Now, we compare the derived secrecy rate with the case where no feedback is
allowed. First let us assume that $n_\bob\ge n_\eve$. Then, in the non-coherent
network coding scenario introduced in \S\ref{sec:SecShar-NonCohNetCode-ProbStatement},
it can be easily verified that the channel from Alice to Eve is a \emph{stochastically 
degraded} (for the definition refer to \cite[p.~373]{LiaPoorShamai-Now09-InfTheSec}) version of the channel from Alice to Bob.

So by applying the result of \cite{Wyner-Bell75-Wiretap} or \cite[Theorem~3]{CsisKor-IT78-BrodConfdMsg},
for the secret key sharing capacity we can write
\begin{align*}
C_s &= \max_{P_{X_\alice}} \left[ I(X_\alice;X_\bob) - I(X_\alice;X_\eve) \right] \nonumber\\
&{=} \max_{P_{\Pi_\alice}} \left[ I(\Pi_\alice; \Pi_\bob) - I(\Pi_\alice;\Pi_\eve) \right],
\end{align*}
where the sufficiency of optimization over subspaces follows from 
a similar argument to \cite[Theorem~1]{JaMoFrDi-IT11}. Similar to the proof of 
Lemma~\ref{lem:MultTerSec-NonCohNetCode-UpperBound},
one can show that
\begin{align*}
 C_s = [n_\bob-n_\eve]^{+} (\ell-n_\bob)\log{q},
\end{align*}
which is positive only if $n_\bob>n_\eve$.
\hfill$\blacksquare$
\end{remark}

The above comparison demonstrates the amount of improvement of
the secret key generation rate we might gain by using feedback.

\subsection{Special Case: Achievability Scheme for Three Terminals}
As an another example, here we consider the three trusted terminals problem
(\ie, $m=2$). As before, we assume that $n_\alice < \ell$ and for the convenience
we suppose that $n_\bob=n_\calvin\le n_\alice$ and $n_\eve\le n_\alice$.

In order to characterize the achievable secrecy rate, we need to find the
dimension of subspaces $U_\bob$, $U_\calvin$, and $U_{\bob\calvin}$ and 
their sums (including $\Pi_\eve$ as well).
We assume that the field size $q$ is large and
we know that $\Pi_\bob$, $\Pi_\calvin$, and $\Pi_\eve$ are chosen uniformly
at random from $\Pi_\alice$.
Subspaces $\Pi_{\bob\calvin}$ and $\Pi_{\bob\eve}$ are also distributed
independently and uniformly at random in $\Pi_\bob$. Similarly, the same
is true for $\Pi_{\bob\calvin}$ and $\Pi_{\calvin\eve}$ in $\Pi_\calvin$.
We have
\begin{align*}
 \left\{\begin{array}{l}
 U_\bob \triangleq \Pi_\bob \setminus_s (\Pi_{\bob\calvin}+\Pi_{\bob\eve}) \\
 U_\calvin \triangleq \Pi_\calvin \setminus_s (\Pi_{\bob\calvin}+\Pi_{\calvin\eve}) \\
 U_{\bob\calvin} \triangleq \Pi_{\bob\calvin} \setminus_s (\Pi_{\bob\calvin\eve}), 
 \end{array}\right.
\end{align*}
so we can write
\par\nobreak{\small
\setlength{\abovedisplayskip}{-2pt}
\begin{align*}
 \dim(U_\bob) &= \dim(\Pi_\bob) -\dim(\Pi_{\bob\calvin}+\Pi_{\bob\eve}) \nonumber\\
 &\stackrel{\text{(a)}}{=} \dim(\Pi_\bob) - \min\left[ \dim(\Pi_{\bob\calvin})+\dim(\Pi_{\bob\eve}), \dim(\Pi_\bob) \right] \nonumber\\
 &\stackrel{\text{(b)}}{=} n_\bob - \min\left[ \dim(\Pi_{\bob\calvin})+\dim(\Pi_{\bob\eve}), n_\bob \right] \nonumber\\
 &= [ n_\bob - \dim(\Pi_{\bob\calvin}) - \dim(\Pi_{\bob\eve}) ]^+ \nonumber\\
 &\stackrel{\text{(c)}}{=} \left[ n_\bob - (2n_\bob - n_\alice)^+ - (n_\bob + n_\eve - n_\alice)^+ \right]^+,
\end{align*}}%
where (a) follows from Lemma~\ref{lem:Uni-and-Joint-k-RandSubSpace} because 
$\Pi_{\bob\calvin}$ and $\Pi_{\bob\eve}$ are chosen independently and 
uniformly at random from $\Pi_\bob$, (b) is true because $q$ is large, 
and (c) follows from Lemma~\ref{lem:Uni-and-Joint-k-RandSubSpace}. Note that because 
we have assumed $n_\bob=n_\calvin$ it follows that $\dim(U_\calvin)=\dim(U_\bob)$.

Similarly, for the dimension of $U_{\bob\calvin}$ we can write
\begin{align*}
 \dim(U_{\bob\calvin}) &= \dim(\Pi_{\bob\calvin}) - \dim(\Pi_{\bob\calvin\eve}) \nonumber\\
 &= \dim(\Pi_{\bob\calvin}) - \left[ \dim(\Pi_{\bob\calvin}) +n_\eve -n_\alice \right]^+ \nonumber\\
 &=  \min \left[ n_\alice- n_\eve, (2n_\bob - n_\alice)^+ \right].
\end{align*}

\begin{proposition}\label{prop:UbUcUbc-Complmtry}
 From the construction, the subspaces $U_\bob$, $U_\calvin$, 
 and $U_{\bob\calvin}$ are orthogonal and similarly the same holds for $U_\bob$, 
 $U_{\bob\calvin}$, and $\Pi_\eve$. Also $U_\calvin$, $U_{\bob\calvin}$,
 and $\Pi_\eve$ are orthogonal w.h.p.
\end{proposition}

Now we may write the linear program stated in Theorem~\ref{thm:NonCoh-Secrecy-AchvThm} 
as follows
\par\nobreak{\small
\setlength{\abovedisplayskip}{-2pt}
\begin{equation*}
\begin{array}{ll}
 \mathrm{maximize} & \min \left[ \theta_{\bob}+\theta_{\bob\calvin}, \theta_{\calvin}+\theta_{\bob\calvin} \right] (\ell - n_\alice)\log{q} \\
\mathrm{subject\ to} & \theta_\bob \le \dim(U_\bob + \Pi_\eve) - n_\eve \\
& \theta_\calvin \le \dim(U_\calvin + \Pi_\eve) - n_\eve \\
& \theta_{\bob\calvin} \le \dim(U_{\bob\calvin} + \Pi_\eve) - n_\eve \\
& \theta_\bob + \theta_\calvin \le \dim(U_\bob + U_\calvin + \Pi_\eve) - n_\eve \\
& \theta_\bob + \theta_\calvin + \theta_{\bob\calvin} \le \dim(U_\bob + U_\calvin + U_{\bob\calvin} + \Pi_\eve) - n_\eve.
\end{array}
\end{equation*}}%
Because of the symmetry in the problem ($n_\bob=n_\calvin$), for the optimal
solution we should have $\theta_\bob=\theta_\calvin$. Knowing this 
and using Proposition~\ref{prop:UbUcUbc-Complmtry}, we may further
simplify the above linear program as follows
\par\nobreak{\small
\setlength{\abovedisplayskip}{-2pt}
\begin{equation*}
\begin{array}{ll}
 \mathrm{maximize} & \left[ \theta_{\bob}+\theta_{\bob\calvin} \right] (\ell - n_\alice)\log{q} \\
\mathrm{subject\ to} & 
 \theta_\bob \le \frac{1}{2} \left[ \dim(U_\bob + U_\calvin + \Pi_\eve) - n_\eve \right] \triangleq \alpha_1 \\
& \theta_{\bob\calvin} \le \dim(U_{\bob\calvin}) \triangleq \alpha_2 \\
& 2\theta_\bob + \theta_{\bob\calvin} \le \dim(U_\bob + U_\calvin + U_{\bob\calvin} + \Pi_\eve) - n_\eve \triangleq \alpha_3.
\end{array}
\end{equation*}}%
From the definitions of $\alpha$'s, we can easily observe that, $\alpha_3\ge 2\alpha_1$, 
$\alpha_3\ge \alpha_2$, and $\alpha_3\le 2\alpha_1+\alpha_2$. Hence, 
$\theta_\bob+\theta_{\bob\calvin}$ gets its maximum at the point 
$(\theta_\bob,\theta_{\bob\calvin})=(\frac{\alpha_3-\alpha_2}{2},\alpha_2)$.
Thus, for the maximum achievable secrecy rate we have
\begin{equation*}
 R_s = \left[ \frac{\alpha_2+\alpha_3}{2} \right] (\ell - n_\alice)\log{q}.
\end{equation*}

As mentioned before, we assume that subspaces $U_\mc{J}$'s are chosen uniformly at random
from $\Pi_\mc{J}$. So $\Pi_\eve$ and $U_\mc{J}$'s are independent and for $\alpha_3$
we can write
\par\nobreak{\small
\setlength{\abovedisplayskip}{-2pt}
\begin{align*}
 \alpha_3 &= \min[\dim(U_\bob)+\dim(U_\calvin)+\dim(U_{\bob\calvin}) + \dim(\Pi_\eve),n_\alice ] - n_\eve \nonumber\\
 &= \min[\dim(U_\bob)+\dim(U_\calvin)+\dim(U_{\bob\calvin}),n_\alice - n_\eve ] \nonumber\\
 &= \min[ 2\dim(U_\bob) + \dim(U_{\bob\calvin}) ,n_\alice - n_\eve ].
\end{align*}}%
So for the secrecy rate (achievable asymptotically when $q$ goes 
to infinity) we have
\par\nobreak{\small
\setlength{\abovedisplayskip}{-2pt}
\begin{gather}
 R_s / (\ell - n_\alice)\log{q}  = \nonumber\\
\min\left[ \dim(U_\bob) + \dim(U_{\bob\calvin}), \frac{1}{2} \left( n_\alice + \dim(U_{\bob\calvin}) - n_\eve \right)  \right]. \label{eq:FinalAchvSecRate-NonCohSecrecy-3Ter}
\end{gather}}%

\begin{example}
As an example, here we compare the achievable secret key sharing rate
among three legitimate terminals (\ie, $m=2$) as derived in
\eqref{eq:FinalAchvSecRate-NonCohSecrecy-3Ter} with the upper bound
stated in Theorem~\ref{thm:MultTerSec-NonCohNetCode-UpperBound}.
We consider two symmetric setup where for the first one we have 
$n_\alice=60$, $n_\bob=n_\calvin=15$ (see Fig.~\ref{fig:Example-NonCohSecrecy-RateReg-na60-nb15})
and for the second one we have $n_\alice=60$, $n_\bob=n_\calvin=45$ 
(see Fig.~\ref{fig:Example-NonCohSecrecy-RateReg-na60-nb45}).
In each of these situations, we depict the upper and lower bounds
on the secret key generation rate as a function of the number of
packets (degrees of freedom) received by Eve.

\begin{figure}[h]
\centering
\subfigure[$m=2$, $n_\alice=60$, and $n_\bob=n_\calvin=15$.]{
\includegraphics[scale=0.47]{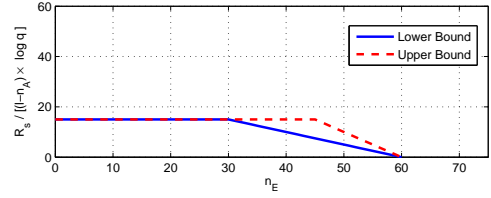}
\label{fig:Example-NonCohSecrecy-RateReg-na60-nb15}
}
\subfigure[$m=2$, $n_\alice=60$, and $n_\bob=n_\calvin=45$.]{
\includegraphics[scale=0.47]{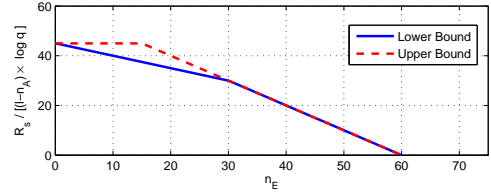}
\label{fig:Example-NonCohSecrecy-RateReg-na60-nb45}
}
\caption{A comparison between the achievable secrecy rate of 
Theorem~\ref{thm:NonCoh-Secrecy-AchvThm} and the upper bound 
given by Theorem~\ref{thm:MultTerSec-NonCohNetCode-UpperBound}
for two cases: (a) when $m=2$, $n_\alice=60$, and $n_\bob=n_\calvin=15$ and
(b) when $m=2$, $n_\alice=60$, and $n_\bob=n_\calvin=45$.}
\end{figure}

\end{example}


\bibliographystyle{IEEEtran}
\bibliography{MyBibliography}

\end{document}